\documentclass{birkjour}
\usepackage[english]{babel}
\usepackage{amssymb}

 \newtheorem{thm}{Theorem}[section]

 \theoremstyle{definition}
 
 \theoremstyle{remark}
 
 \newtheorem*{ex}{Example}
 \numberwithin{equation}{section}

\def\cl{{C}\!\ell}
\def\Cen{{\rm Cen}}
\def\diag{{\rm diag}}
\def\Mat{{\rm Mat}}
\def\GL{{\rm GL}}
\def\U{{\rm U}}
\def\Or{{\rm O}}
\def\R{{\mathbb R}}
\def\C{{\mathbb C}}
\def\Tr{{\rm Tr}}
\def\G{{\rm G}}
\def\T{{\rm T}}
\def\M{{\rm M}}

\begin{document}

%
%
%
%
%
%
%
%
%

\title[Covariantly constant solutions of the Yang-Mills equations]
 {Covariantly constant solutions \\ of the Yang-Mills equations}

\author[Dmitry Shirokov]{Dmitry Shirokov}

\address{%
National Research University Higher School of Economics\\
Myasnitskaya str. 20\\
101000 Moscow\\
Russia}

\email{dshirokov@hse.ru}

\thanks{The reported study was funded by RFBR according to the research project
No. 16-31-00347 mol\_a.
}

\address{ Institute for Information Transmission Problems of Russian Academy of Sciences\\
Bolshoy Karetny per. 19 \\
127051 Moscow \\
Russia}

\email{shirokov@iitp.ru}


\date{September 22, 2017}



\begin{abstract}
We present a new class of covariantly constant solutions of the Yang-Mills equations. These solutions correspond to the solution of the field equation for the spin connection of the general form.
\end{abstract}

\maketitle

\section*{Introduction}

This paper is based on the 40-min talk given by the author at the International Conference on Clifford Algebras and Their Applications in Mathematical Physics (ICCA 11) in the mini-symposium ``Clifford Algebra and the Fundamental Forces of Nature'' (Ghent, August 2017).

We present a new class of covariantly constant solutions of the Yang-Mills equations. These solutions correspond to the solution of the field equation for the spin connection of the general form \cite{ofe}.

In the present paper, we generalize results of the papers \cite{marchuk1} and \cite{ofe}. Namely, Theorems \ref{th5} and \ref{th6} generalize the main results of the paper \cite{ofe} on the spin connection of the general form. Theorems \ref{CCS}, \ref{CCS2}, and \ref{CCS3} generalize the main results of the paper \cite{marchuk1} on a new class of solutions of the Yang-Mills equations.

\section{Tensor fields with values in Clifford algebra and field equation for spin connection}\label{sec1}

Let us consider pseudo-Euclidean space $\R^{k,l}$ of the dimension $\dim \R^{k,l}=k+l=m$ with Cartesian coordinates $x^\mu$, $\mu=1, \ldots, m$. The metric tensor of $\R^{k,l}$ is given by a diagonal matrix
\begin{eqnarray}
\rho=||\rho^{\mu\mu}||=\diag(\underbrace{1, \ldots, 1}_{k}, \underbrace{-1, \ldots, -1}_{l}).\label{rho}
\end{eqnarray}
By $\partial_\mu:=\frac{\partial}{\partial x^\mu}$ denote the partial derivatives.

Let us consider the real Clifford algebra $\cl_{p,q}$, $p+q=n$ \cite{Lounesto}, \cite{Benn:Tucker}, with the generators $e^a$, $a=1, \ldots, n$, which satisfy the following conditions
\begin{eqnarray}
e^a e^b+e^b e^a=2\eta^{ab}e,\qquad \eta=||\eta^{ab}||=\diag(\underbrace{1, \ldots, 1}_{p}, \underbrace{-1, \ldots, -1}_{q}),\label{eta}
\end{eqnarray}
and the basis
$$\{e^A\}=\{e, e^a, e^{ab}, \ldots, e^{1\ldots n}\},$$
where $A$ is an arbitrary ordered multi-index of a length $|A|$ between $0$ and $n$.

We use notation $\cl_{p,q}\T^r_s$ for the set of tensor fields with values in the Clifford algebra
$U^{\phi_1 \ldots \phi_r}_{\psi_1 \ldots \psi_s}=U^{\phi_1 \ldots \phi_r}_{\psi_1 \ldots \psi_s}(x): \R^{k,l}\to\cl_{p,q}$:
$$U^\Phi_\Psi (x)=U^{\phi_1 \ldots \phi_r}_{\psi_1 \ldots \psi_s}(x)=u^{\phi_1 \ldots \phi_r}_{\psi_1 \ldots \psi_s}(x)e+u^{\phi_1 \ldots \phi_r}_{\psi_1 \ldots \psi_s a}(x)e^a +u^{\phi_1 \ldots \phi_r}_{\psi_1 \ldots \psi_s ab}(x)e^{ab}+\cdots$$
$$\cdots+u^{\phi_1 \ldots \phi_r}_{\psi_1 \ldots \psi_s 1\ldots n}(x)e^{1\ldots n}=u^\Phi_{\Psi A}(x)e^A\in\cl_{p,q}\T^r_s,\qquad u^\Phi_{\Psi A}:\R^{k,l}\to\R.$$
We denote multi-index $\phi_1 \ldots \phi_r$ by $\Phi$, multi-index $\psi_1 \ldots \psi_s$ by $\Psi$ , and their lengths by $|\Phi|=r$, $|\Psi|=s$.

We can raise and lower Greek indices using matrix $\rho=||\rho^{\mu\nu}||$ and raise and lower Latin indices using matrix $\eta=||\eta^{ab}||$. We have $e_a=\eta_{ab}e^b=(e^a)^{-1}$ and $e_A=(e^A)^{-1}$.

As a particular case, we consider the algebra $\cl_{p,q}\T$ of smooth functions with values in the Clifford algebra $U=U(x):\R^{k,l}\to\cl_{p,q}$:
$$U(x)=u(x)e+u_a(x) e^a+\cdots+u_{1\ldots n}(x)e^{1\ldots n}=u_A(x) e^A,\quad u_A: \R^{k,l}\to \R.$$

Let us consider a set of smooth functions with values in Clifford algebra $h^a: \R^{k,l}\to\cl_{p,q}$
\begin{eqnarray}
h^a(x)=y^a(x)e+y^{a}_b(x) e^b+\cdots+ y^{a}_{1 \ldots n}(x)e^{1\ldots n}=y^a_A(x)e^A,\label{he}
\end{eqnarray}
which satisfy conditions
\begin{eqnarray}
h^a(x) h^b(x) +h^b(x) h^a(x)=2\eta^{ab}e,\qquad a, b=1, \ldots, n,\qquad \forall x \in \R^{k,l}.\label{hh}
\end{eqnarray}
In the case of odd $n=p+q$, we also require the following additional condition to obtain independent elements $h^A$ (see \cite{Shirokov}):
\begin{eqnarray}
\Tr(h^{1}\cdots h^n)=0\qquad \mbox{(in the case of odd $n$)}\label{trh}
\end{eqnarray}
where $\Tr: \cl_{p,q} \to \cl^0_{p,q}$ is the projection operator onto subspace $\cl^0_{p,q}=\{ue\}$ of grade $0$.

The set
$$\{h^A(x)\}=\{e, h^a(x), \ldots, h^{1\ldots n}(x)\}$$
is a basis of the algebra $\cl_{p,q}\T$ of smooth functions with values in the Clifford algebra.

The Clifford algebra $\cl_{p,q}$ is a Lie algebra with respect to the commutator $[U, V]=UV-VU$. Let us consider the following subset $$\cl^\circledS_{p,q}=\cl_{p,q}\setminus\Cen(\cl_{p,q})$$
which is a Lie subalgebra of $\cl_{p,q}$. Note that the center \cite{Lounesto} of $\cl_{p,q}$ is
\begin{eqnarray}
\Cen(\cl_{p,q})=\left\{\begin{array}{ll}
\cl^0_{p,q}, & \hbox{if $n$ is even;} \\
\cl^0_{p,q}\oplus\cl^n_{p,q}, & \hbox{if $n$ is odd,}
\end{array}
\right.\label{center}
\end{eqnarray}
where
$$\cl^j_{p,q}:=\{ \sum_{A: |A|=j} u_A e^A \},\qquad j=0, 1, \ldots, n,$$
is the subspace of $\cl_{p,q}$ of grade $j$.

\begin{thm}\label{th1} For elements $h^a\in\cl_{p,q}\T$ (\ref{he}), which satisfy conditions (\ref{hh}) and (\ref{trh}), we have
$$h^a\in \cl^\circledS_{p,q}\T,\qquad a=1, \ldots, n.$$
\end{thm}

Let us consider the following system of equations for unknown elements $C_\mu\in\cl_{p,q}\T_1$
\begin{eqnarray}
\partial_\mu h^a-[C_\mu, h^a]=0,\qquad \mu=1, \ldots, m,\qquad a=1, \ldots, n.
\label{prim}
\end{eqnarray}

It is convenient to consider $C_\mu \in \cl^\circledS_{p,q}\T_1$ because of the commutator in (\ref{prim}). So, let all expressions in equation (\ref{prim}) belong to the Lie algebra $\cl^\circledS_{p,q}$.
We denote the group of all invertible Clifford algebra elements by $\cl^\times_{p,q}$.

\begin{thm}\label{th2}
Let $S : \R^{k,l}\to\cl^\times_{p,q}$ be a function with values in $\cl^\times_{p,q}$ such that
$$
S^{-1}\partial_\mu S\in\cl^\circledS_{p,q}\T_1.
$$
Then the following expressions
$$
\acute h^a=S^{-1}h^a S\in\cl^\circledS_{p,q}\T,\quad
\acute C_\mu=S^{-1}C_\mu S-S^{-1}\partial_\mu S\in\cl^\circledS_{p,q}\T_1
$$
also satisfy the equation
$$
\partial_\mu \acute h^a-[\acute C_\mu,\acute h^a]=0,\quad \forall \mu=1,\ldots, m,\qquad a=1, \ldots, n.
$$
\end{thm}

\begin{thm}\label{th3}
The system (\ref{prim}) has unique solution $C_\mu\in\cl^\circledS_{p,q}\T_1$:
\begin{eqnarray}
C_\mu=
\left\{\begin{aligned}
&\sum_{i=1}^{n} \mu_i \pi\{h\}_i ((\partial_\mu h^a) h_a),\quad \mu=1, \ldots, m,\qquad  \hbox{if $n$ is even;} \\
&\sum_{i=1}^{\frac{n-1}{2}} \mu_i \pi\{h\}_{i, n-i} ((\partial_\mu h^a) h_a),\quad \mu=1, \ldots, m,\qquad \hbox{if $n$ is odd,}
\end{aligned}
\right.\label{Cha}
\end{eqnarray}
where
$$\mu_i=\frac{1}{n-\lambda_i},\qquad \lambda_i=(-1)^i(n-2i),$$
and
$$\pi\{h\}_i: \cl_{p,q}\to \cl\{h\}^i_{p,q}=\{\sum_{A: |A|=i}u_{A} h^A\},\quad \pi\{h\}_{i, n-i}=\pi\{h\}_i+\pi\{h\}_{n-i}$$
are projection operators.

We also have
$$\pi\{h\}_i(U)=\sum_{j=0}^n b_{ij}F^j(U),\qquad \pi\{h\}_{i, n-i}(U)=\sum_{j=0}^{\frac{n-1}{2}} g_{ij}F^j(U),\quad \mbox{where}$$
$$B_{n+1}=||b_{ij}||=A_{n+1}^{-1},\, a_{ij}=(\lambda_{j-1})^{i-1},\quad G_{\frac{n-1}{2}}=D_{\frac{n-1}{2}}^{-1},\, d_{ij}=(\lambda_{j-1})^{i-1},$$
$$F^j(U)=\underbrace{F(F(\cdots F(}_{j}U))\cdots),\qquad F(U)=\sum_{a=1}^n h_a U h^a.$$
\end{thm}

\begin{thm}\label{th4} The following condition of zero-curvature follows from (\ref{prim}):
\begin{eqnarray}
\partial_\mu C_\nu-\partial_\nu C_\mu-[C_\mu, C_\nu]=0,\qquad \mu, \nu=1, \ldots, m.\label{CC}
\end{eqnarray}
\end{thm}

\begin{proof} (of Theorems \ref{th1} - \ref{th4}). Analogous theorems are proved in \cite{ofe} for the vector fields $h^\mu\in \cl_{p,q}\T^1$ in pseudo-Euclidean space $\R^{p,q}$ (w.r.t. the orthogonal transformations of coordinates), which we also discuss in Section \ref{sec5} of the current paper. But all these statements can be also proved for the scalar functions $h^a\in \cl_{p,q}\T$ in pseudo-Euclidean space $\R^{k,l}$, $k\neq p$, $l\neq q$ (it is sufficient to change all $h^\mu$ to $h^a$ in the proof of the corresponding theorems in \cite{ofe}).
\end{proof}

\section{Spin connection of the general form}\label{sec2}

We use the method of averaging in Clifford algebra \cite{averaging} to obtain another form of unique solution of the system (\ref{prim}).

\begin{thm}\label{th5} From the system (\ref{prim}) it follows that
\begin{eqnarray}
\partial_\mu h^A-[C_\mu, h^A]=0,\qquad \mu=1, \ldots, m\label{hA}
\end{eqnarray}
for all ordered multi-indices $A$ of a length between $0$ and $n$.
\end{thm}
\begin{proof}
Let us multiply both sides of (\ref{prim}) on the left and on the right by the required number of elements $h^{a_j}$:
\begin{eqnarray}
&&\partial_\mu(h^{a_1})h^{a_2} \cdots h^{a_j}-C_\mu h^{a_1}h^{a_2} \cdots h^{a_j}+h^{a_1}C_\mu h^{a_2}\cdots h^{a_j}=0,\nonumber\\
&&h^{a_1}\partial_\mu (h^{a_2}) h^{a_3}\cdots h^{a_j}-h^{a_1}C_\mu h^{a_2}h^{a_3}\cdots h^{a_j}+ h^{a_1}h^{a_2}C_\mu h^{a_3}\cdots h^{a_j}=0,\nonumber\\
&&\cdots\nonumber\\
&&h^{a_1}\cdots h^{a_{j-1}}\partial_\mu(h^{a_j})-h^{a_1}\cdots h^{a_{j-1}}C_\mu h^{a_j}+h^{a_1}\cdots h^{a_{j-1}}h^{a_j}C_\mu=0.\nonumber
\end{eqnarray}
Summing these equations, we get
$$\partial_\mu(h^{a_1}\cdots h^{a_j})-[C_\mu, h^{a_1}\cdots h^{a_j}]=0$$
for arbitrary indices $a_1, \ldots, a_j$. We obtain (\ref{hA}) for all ordered multi-indices of a length between $0$ and $n$. In the case of empty multi-index, we have
$$\partial_\mu(e)-[C_\mu, e]=0$$
because $e$ does not depend on $x$ and lies in the center $\Cen(\cl_{p,q})$ of the Clifford algebra.
\end{proof}

\begin{thm}\label{th6} The system (\ref{prim})
has unique solution $C_\mu\in\cl^\circledS_{p,q}\T_1$
\begin{eqnarray}
C_\mu=\frac{1}{2^n}(\partial_\mu h^A) h_A,\qquad \mu=1, \ldots, m.\label{ChA}
\end{eqnarray}
In the case of odd $n$, the expression (\ref{ChA}) can be represented as
\begin{eqnarray}
C_\mu=\frac{1}{2^{n-1}}\sum_{|A|=1}^{\frac{n-1}{2}}(\partial_\mu h^A) h_A,\qquad \mu=1, \ldots, m.\label{ChA2}
\end{eqnarray}
\end{thm}
\begin{proof} Let us multiply both sides of (\ref{hA}) by $h_A$ on the right:
\begin{eqnarray}
(\partial_\mu h^A)h_A-C_\mu h^A h_A+h^A C_\mu h_A=0.\label{66}
\end{eqnarray}
Since $h^A h_A=2^n e$ and $h^A C_\mu h_A=2^n \pi_{\Cen}(C_\mu)$  (see \cite{averaging}), we get
$$2^n (C_\mu-\pi_{\Cen}(C_\mu))=(\partial_\mu h^A)h_A.$$
Here we denote the projection operator onto the center of Clifford algebra (\ref{center}) by
$$\pi_{\Cen}:\cl_{p,q}\to\Cen(\cl_{p,q}).$$
Using $C_\mu\in \cl^\circledS_{p,q}\T_1$, we obtain (\ref{ChA}).
By Theorem \ref{th3} the solution $C_\mu\in\cl^\circledS_{p,q}\T_1$ of the system (\ref{prim}) is unique.

In the case of odd $n$, the element $h^{1\ldots n}$ does not depend on $x$ and equals $\pm e^{1\ldots n}$ (see \cite{Shirokov}). Using $h^{1\ldots n}=\pm e^{1\ldots n}\in\Cen(\cl_{p,q})$, we get
$$\sum_{A: |A|=j}(\partial_\mu h^A) h_A=h^{1\ldots n}\sum_{A: |A|=j}(\partial_\mu h^A) h_A h_{1\ldots n}=\sum_{A: |A|=n-j}(\partial_\mu h^A) h_A.$$
Finally, we have
$$(\partial_\mu h^A) h_A=\sum_{|A|=0}^n(\partial_\mu h^A) h_A=2\sum_{|A|=1}^{\frac{n-1}{2}}(\partial_\mu h^A) h_A,\qquad \mu=1, \ldots, m$$
and (\ref{ChA2}). The theorem is proved.
\end{proof}

\begin{ex}
Let us consider the case $n=2$. Using (\ref{Cha}), we get (see \cite{ofe})
\begin{eqnarray}
C_\mu&=&\frac{1}{2}\pi\{h\}_1((\partial_\mu h^a)h_a)+\frac{1}{4}\pi\{h\}_2((\partial_\mu h^a)h_a)\nonumber\\
&=&\frac{1}{2}(\partial_\mu h^a) h_a-\frac{1}{16}h^b (\partial_\mu h^a)h_a h_b-\frac{3}{32}h^c h^b (\partial_\mu h^a)h_a h_b h_c.\label{yy6}
\end{eqnarray}
Let us show that this expression coincides with (\ref{ChA}). We have
\begin{eqnarray}
&&h^b (\partial_\mu h^a)h_a h_b=-(\partial_\mu h^1)h_1+h^1 (\partial_\mu h^2)h_2 h_1+h^2 (\partial_\mu h^1)h_1 h_2-(\partial_\mu h^2)h_2,\nonumber\\
&&h^c h^b (\partial_\mu h^a)h_a h_b h_c=(\partial_\mu h^1)h_1+(\partial_\mu h^2)h_2+(\partial_\mu h^2)h_2-h^2 (\partial_\mu h^1) h_1 h_2 \nonumber\\
&&-h^1 (\partial_\mu h^2)h_2 h_1-h^2 (\partial_\mu h^1) h_1 h_2-h^1 (\partial_\mu h^2)h_2 h_1 +(\partial_\mu h^1)h_1.\nonumber
\end{eqnarray}
Finally, we obtain from (\ref{yy6})
\begin{eqnarray}
C_\mu&=&\frac{3}{8}(\partial_\mu h^a) h_a+\frac{1}{8}(h^1 (\partial_\mu h^2)h_2 h_1+h^2 (\partial_\mu h^1) h_1 h_2)\nonumber\\
&=&\frac{1}{4}(\partial_\mu h^a) h_a+\frac{1}{8}(\partial_\mu (h^{1}h^2)h_{2}h_1+\partial_\mu (h^{2} h^1)h_{1}h_2)\nonumber\\
&=&\frac{1}{4}(\partial_\mu h^a) h_a+\frac{1}{4}(\partial_\mu h^{12})h_{12}=\frac{1}{4}(\partial_\mu h^A)h_A.\nonumber
\end{eqnarray}
Also we can verify that
$$C_\mu=\frac{1}{4}(\partial_\mu h^A)h_A=\frac{1}{4}(\partial_\mu h^1) h_1+\frac{1}{4}(\partial_\mu h^2) h_2+\frac{1}{4}(\partial_\mu h^{12})h_{12}$$
is a solution of (\ref{prim}). We have
\begin{eqnarray}
&&[C_\mu, h^1]=\frac{1}{4}(\partial_\mu h^1) h_1 h^1+\frac{1}{4}(\partial_\mu h^2) h_2 h^1+\frac{1}{4}(\partial_\mu h^{1})h^2 h_2 h_1 h^1\nonumber\\
&&+\frac{1}{4}h^1 (\partial_\mu h^2) h_2 h_1 h^1-\frac{1}{4}h^1(\partial_\mu h^1) h_1-\frac{1}{4}h^1(\partial_\mu h^2) h_2-\frac{1}{4}h^1(\partial_\mu h^{1})h^2h_{2}h_1\nonumber\\
&&-\frac{1}{4}h^1 h^1 (\partial_\mu h^2) h_2 h_1=\frac{1}{4}\partial_\mu h^1+\frac{1}{4}(\partial_\mu h^2) h_2 h^1+\frac{1}{4}\partial_\mu h^{1}+\frac{1}{4}h^1 (\partial_\mu h^2) h_2\nonumber\\
&&+\frac{1}{4}\partial_\mu h^1-\frac{1}{4}h^1(\partial_\mu h^2) h_2+\frac{1}{4}\partial_\mu h^{1}-\frac{1}{4}(\partial_\mu h^2) h_2 h^1=\partial_\mu h^1.\nonumber
\end{eqnarray}
We can similarly verify $[C_\mu, h^2]=\partial_\mu h^2$.
\end{ex}

\begin{ex}
Let us consider the case $n=3$. Since (\ref{Cha}), we have
\begin{eqnarray}
C_\mu=\frac{1}{4}\pi\{h\}_{1, 2}((\partial_\mu h^a)h_a)=\frac{3}{16}(\partial_\mu h^a)h_a-\frac{1}{16}h^b(\partial_\mu h^a) h_a h_b.\label{nov}
\end{eqnarray}
Using (\ref{ChA}), we get
$$C_\mu=\frac{1}{8}(\partial_\mu h^A)h_A=\frac{1}{8}(\partial_\mu h^a) h_a+\frac{1}{8}(\partial_\mu h^{ab}) h_{ab},$$
because $h^{123}$ does not depend on $x$ and equals $\pm e^{123}$ (see \cite{Shirokov}).
Since $h^{123}=\pm e^{123}\in\Cen(\cl_{p,q})$, we obtain
$$(\partial_\mu h^a) h_a=h^{123}(\partial_\mu h^a) h_a h_{123}=(\partial_\mu h^{ab}) h_{ab}.$$
Finally, we have in the case $n=3$
\begin{eqnarray}
C_\mu=\frac{1}{4}(\partial_\mu h^a) h_a\in(\cl^1_{p,q}\oplus\cl^2_{p,q})\T_1.\label{nov2}
\end{eqnarray}
Let us show that (\ref{nov2}) coincides with (\ref{nov}). We have
\begin{eqnarray}
&&X:=\partial_\mu(h^{ab})h_{ab}=\frac{1}{2}\partial_\mu(h^a h^b)h_b h_a=\frac{1}{2}(\partial_\mu h^a) h^b h_b h_a+\frac{1}{2}h^a (\partial_\mu h^b) h_b h_a\nonumber\\
&&=\frac{3}{2}(\partial_\mu h^a)h_a+\frac{1}{2}h^a (\partial_\mu h^b) h_b h_a=\frac{3}{2}X+\frac{1}{2}h^a (\partial_\mu h^b) h_b h_a.\nonumber
\end{eqnarray}
Then
$$h^a (\partial_\mu h^b) h_b h_a=-X=-\partial_\mu(h^{ab})h_{ab}=-(\partial_\mu h^a) h_a.$$
Using this identity we get (\ref{nov2}) from (\ref{nov}).
\end{ex}

\begin{ex} Let us consider the case $h^a\in\cl^1_{p,q}\T$. We have (instead of a general case (\ref{he}))
\begin{eqnarray}
h^a(x)=y^a_b(x) e^b\in\cl^1_{p,q}\T.\label{hy}
\end{eqnarray}
Using (\ref{hh}), we get the following conditions for the elements $y^a_b$:
\begin{eqnarray}
y^a_b y^c_d \eta^{bd}=\eta^{ac},\label{yy}
\end{eqnarray}
which are orthogonality conditions for the matrix $Y=||y^a_b||$:
$$Y\in\Or(p,q)=\{Y\in\Mat(n, \R), Y^\T \eta Y=\eta\}.$$

In this particular case, we have $\partial_\mu h^a\in\cl^1_{p,q}\T_1$.
It can be proved that (see \cite{marchuk1})
$$(\partial_\mu h^a) h_a=(\partial_\mu h_a) h^a=-h_a(\partial_\mu h^a)= -h^a(\partial_\mu h_a)\in\cl^2_{p,q}\T_1$$
and the solution of (\ref{prim}) will be
\begin{eqnarray}
C_\mu=\frac{1}{4}(\partial_\mu h^a) h_a\in\cl^2_{p,q}\T_1,\label{spincon}
\end{eqnarray}
in the case of arbitrary $n$. This expression is known as \emph{spin connection}.
Since (\ref{hy}), we have
\begin{eqnarray}
C_\mu=\frac{1}{4}\eta_{ac}\partial_\mu(y^a_b)y^c_d e^b e^d.\label{Comega}
\end{eqnarray}
Using (\ref{yy}), we get
$$\partial_\mu(y^a_b) y^c_d \eta^{bd}+y^a_b \partial_\mu(y^c_d) \eta^{bd}=0$$
and can rewrite (\ref{Comega}) in the following way
$$C_\mu=\omega_{\mu bd}e^{bd}\in\cl^2_{p,q}\T_1,\qquad \omega_{\mu bd}=\frac{1}{2}\sum_{b<d}\eta_{ac}\partial_\mu(y^a_b)y^c_d.$$

The expression (\ref{spincon}) coincides with (\ref{Cha}) and (\ref{ChA}) in the case $h^a\in\cl^1_{p,q}\T$. Note that the expressions (\ref{Cha}) and (\ref{ChA}) do not belong to $\cl^2_{p,q}\T_1$ in the general case and we call them \emph{spin connection of the general form}.
\end{ex}

\section{Covariant derivatives and covariantly constant tensor fields with values in Clifford algebra.}\label{sec3}

Let us consider a tensor field with values in the Clifford algebra $U^{\Phi}_{\Psi}=U^{\phi_1 \ldots \phi_r}_{\psi_1 \ldots \psi_s}(x):\R^{k,l}\to\cl_{p,q}$, $|\Phi|=r$, $|\Psi|=s$:
$$U^{\Phi}_{\Psi}(x)=u^{\Phi}_{\Psi A}(x)e^A\in\cl_{p,q}\T^r_s,\qquad u^{\Phi}_{\Psi A}(x):\R^{k,l}\to\R.$$
We can take $h^b(x)=y^b_A(x) e^A$ (see (\ref{he}), (\ref{hh}), and (\ref{trh})) and obtain another basis $h^B(x)=y^B_A (x) e^A$ of $\cl_{p,q}\T$ for some $y^B_A=y^B_A(x):\R^{k,l}\to\R$. We have
$$U^{\Phi}_{\Psi}(x)=u\{h\}^{\Phi}_{\Psi B}(x)h^B(x),\qquad u\{h\}^{\Phi}_{\Psi B}(x):\R^{k,l}\to\R,$$
where
$$u^\Phi_{\Psi A}(x)=u\{h\}^\Phi_{\Psi B}(x) y^B_A(x).$$
Let us consider the following \emph{operation of covariant differentiation} that depends on the basis $\{h^A\}$ of $\cl_{p,q}\T$
\begin{eqnarray}
D_\mu U^\Phi_{\Psi}:=\partial_\mu U^\Phi_{\Psi}-[C_\mu, U^\Phi_{\Psi}],\qquad U^\Phi_\Psi\in\cl_{p,q}\T^r_s,\label{covder}
\end{eqnarray}
where $C_\mu=C_\mu (x)\in\cl^\circledS_{p,q}\T_1$ is a unique solution of (\ref{prim}).

\begin{thm}\label{th7} For arbitrary tensor field with values in $\cl_{p,q}$
$$U^{\Phi}_{\Psi}(x)=u\{h\}^{\Phi}_{\Psi B}(x)h^B(x)\in\cl_{p,q}\T^r_s$$ we have
\begin{eqnarray}
D_\mu (U^{\Phi}_{\Psi}(x))= \partial_\mu (u\{h\}^{\Phi}_{\Psi B}(x)) h^B(x).\label{covder2}
\end{eqnarray}
\end{thm}
\begin{proof} Using Theorem \ref{th5} and (\ref{covder}), we get
\begin{eqnarray}
&&D_\mu (U^{\Phi}_{\Psi})=D_\mu(u\{h\}^{\Phi}_{\Psi B}h^B)=\partial_\mu(u\{h\}^{\Phi}_{\Psi B}h^B)-[C_\mu, u\{h\}^{\Phi}_{\Psi B}h^B]\nonumber\\
&&=\partial_\mu(u\{h\}^{\Phi}_{\Psi B})h^B+u\{h\}^{\Phi}_{\Psi B}\partial_\mu(h^B)-u\{h\}^{\Phi}_{\Psi B}[C_\mu, h^B]\nonumber\\
&&=\partial_\mu(u\{h\}^{\Phi}_{\Psi B})h^B+u\{h\}^{\Phi}_{\Psi B}(\partial_\mu(h^B)-[C_\mu, h^B])\nonumber\\
&&=\partial_\mu(u\{h\}^{\Phi}_{\Psi B})h^B.\nonumber
\end{eqnarray}
The theorem is proved.
\end{proof}

Let us consider a set of \emph{covariantly constant tensor fields with values in Clifford algebra} $U^\Phi_\Psi\in\cl_{p,q}\T^r_s$
$$\M\cl_{p,q}\T^r_s:=\{U^\Phi_\Psi\in\cl_{p,q}\T^r_s,\quad D_\mu U^\Phi_\Psi=0\}.$$

Elements of this set have the form
$$U^{\phi_1 \ldots \phi_r}_{\psi_1 \ldots \psi_s}=u^{\phi_1 \ldots \phi_r}_{\psi_1 \ldots \psi_s} e+u^{\phi_1 \ldots \phi_r}_{\psi_1 \ldots \psi_s a_1} h^a_1+ u^{\phi_1 \ldots \phi_r}_{\psi_1 \ldots \psi_s a_1 a_2}h^{\nu_1 \nu_2}+\cdots+ u^{\phi_1 \ldots \phi_r}_{\psi_1 \ldots \psi_s 1 \ldots n}h^{1\ldots n},$$
where all $u^{\Phi}_{\Psi A}$ do not depend on $x$.

\begin{thm}\label{th8} The operation (\ref{covder}) has the following properties:
\begin{eqnarray}
D_\mu(U^\Phi_{\Psi}+V^\Phi_\Psi)=D_\mu(U^\Phi_\Psi)+D_\mu(V^\Phi_\Psi),\qquad D_\mu(\lambda U^\Phi_\Psi)=\lambda D_\mu(U^\Phi_\Psi),\quad \lambda\in\R,\nonumber\\
D_\mu(U^\Phi_\Psi W^\Omega_\Sigma)=U^\Phi_\Psi D_\mu(W^\Omega_\Sigma)+D_\mu(U^\Phi_\Psi) W^\Omega_\Sigma,\qquad D_\mu (D_\nu(U^\Phi_\Psi))=D_\nu (D_\mu(U^\Phi_\Psi))\nonumber
\end{eqnarray}
for arbitrary tensor fields with values in the Clifford algebra $U^\Phi_{\Psi}, V^\Phi_\Psi\in\cl_{p,q}\T^{|\Phi|}_{|\Psi|}$ and $W^\Omega_\Sigma\in\cl_{p,q}\T^{|\Omega|}_{|\Sigma|}$.
\end{thm}
\begin{proof} We can easily obtain these properties of the operation $D_\mu$ in two different ways: using (\ref{covder}) or using (\ref{covder2}). If we use definition (\ref{covder}), then we need also (\ref{CC}) to prove the last property.
\end{proof}

\begin{thm}\label{th9} We have
\begin{eqnarray}
&&\partial_\lambda[C_\mu, C_\nu]+\partial_\mu[C_\nu, C_\lambda]+\partial_\nu[C_\lambda, C_\mu]=0,\nonumber\\
&&D_\lambda[C_\mu, C_\nu]+D_\mu[C_\nu, C_\lambda]+D_\nu[C_\lambda, C_\mu]=0,\nonumber\\
&&D_\mu (C_\rho)=\partial_\rho C_\mu.\nonumber
\end{eqnarray}
\end{thm}

\begin{proof} The first two properties follow from (\ref{CC}) and
$$D_\mu (D_\nu(U^\Phi_\Psi))=D_\nu (D_\mu(U^\Phi_\Psi)).$$
Using (\ref{CC}), we get
$$D_\mu (C_\rho)=\partial_\mu C_\rho-[C_\mu, C_\rho]=\partial_\rho C_\mu.$$
The theorem is proved.
\end{proof}

Note that covariant and partial derivatives do not commute:
$$\partial_\nu D_\mu U^\Phi_\Psi=D_\mu \partial_\nu U^\Phi_\Psi+[U^\Phi_\Psi,\partial_\nu C_\mu],\qquad U^\Phi_\Psi\in\cl_{p,q}\T^r_s.$$

\section{Class of covariantly constant solutions of the Yang-Mills equations}\label{sec4}

Let $\G$ be a semisimple Lie group and $\frak{g}$ be the real Lie algebra of the Lie group $\G$. Multiplication of elements of $\frak{g}$ is given by the Lie bracket  $[U, V]=-[V, U]$. By $\frak{g}\T^r_s$ we denote a set of tensor fields of the pseudo-Euclidean space $\R^{k,l}$, $k+l=m$, of type $(r,s)$ and of rank $r+s$ with values in the Lie algebra $\frak{g}$.

Consider the following equations in the pseudo-Euclidean space $\R^{k,l}$:
\begin{eqnarray}
&&\partial_\mu B_\nu-\partial_\nu B_\mu-[B_\mu, B_\nu]=F_{\mu\nu},\qquad \mu, \nu=1, \ldots, m,\label{YM}\\
&&\partial_\mu F^{\mu\nu}-[B_\mu,F^{\mu\nu}]=J^\nu,\qquad \nu=1, \ldots, m,\nonumber
\end{eqnarray}
where  $B_\mu\in \frak{g}\T_1$, $J^\nu\in\frak{g}\T^1$, $F_{\mu\nu}=-F_{\nu\mu}\in \frak{g}\T_{2}$. These equations are called {\em the Yang-Mills equations} (system  of Yang-Mills equations).
One suggests that $B_\mu, F_{\mu\nu}$ are unknown and $J^\nu$ is known vector with values in the Lie algebra $\frak{g}$. One says that equations (\ref{YM}) define {\em the Yang-Mills field} $(B_\mu, F_{\mu\nu})$, where $B_\mu$ is {\em the potential} and $F_{\mu\nu}$ is {\em the strength} of the Yang-Mills field.
A vector $J^\nu$ is called {\em a non-Abelian current} (in the case of Abelian group $\G$ vector  $J^\nu$ is called {\em a current}).

Consider $B_\mu, F_{\mu\nu}, J^\nu$ that satisfy (\ref{YM}). Let us take a scalar field $S=S(x)$ with values in the Lie group $\G$ and consider transformed tensor fields
\begin{eqnarray}
\acute B_\mu &=& S^{-1}B_\mu S-S^{-1}\partial_\mu S,\nonumber\\
\acute F_{\mu\nu} &=& S^{-1}F_{\mu\nu}S,\label{gauge:tr}\\
\acute J^{\nu} &=& S^{-1}J^{\nu}S.\nonumber
\end{eqnarray}
These tensor fields satisfy the same Yang-Mills equations
\begin{eqnarray*}
&&\partial_\mu\acute B_\nu-\partial_\nu \acute B_\mu-[\acute B_\mu,\acute B_\nu]=\acute F_{\mu\nu},\qquad \mu, \nu=1, \ldots, m,\\
&&\partial_\mu\acute F^{\mu\nu}-[\acute B_\mu,\acute F^{\mu\nu}]=\acute J^\nu,\qquad \nu=1, \ldots, m,
\end{eqnarray*}
i.e., equations (\ref{YM}) are invariant w.r.t. the transformations (\ref{gauge:tr}). The transformation (\ref{gauge:tr}) is called {\em a gauge transformation} (or {\em a  gauge symmetry}), and the Lie group $\G$ is called {\em the gauge group} of the Yang-Mills equations (\ref{YM}).

Let $B_\mu\in \frak{g}\T_1$ be an arbitrary covector with values in $\frak{g}$, which smoothly depends on $x\in\R^{p,q}$. By $F_{\mu\nu}$ denote the expression
\begin{equation}
F_{\mu\nu}:=\partial_\mu B_\nu-\partial_\nu B_\mu-[B_\mu,B_\nu]\label{Fmunu}
\end{equation}
and by $J^\nu$ denote the expression
$$
J^\nu:=\partial_\mu F^{\mu\nu}-[B_\mu,F^{\mu\nu}].
$$
Now we can consider the expression  $\partial_\nu J^\nu-[B_\nu,J^\nu]$ and, with the aid of simple calculations, we may verify that
\begin{equation}
\partial_\nu J^\nu-[B_\nu,J^\nu]=0.\label{nonabel:conslaw}
\end{equation}
This identity is called {\em a non-Abelian conservation law} (in the case of an Abelian Lie group $\G$ we have $\partial_\nu J^\nu=0$, i.e., the divergence of the vector $J^\nu$ equals zero). Therefore the non-Abelian conservation law (\ref{nonabel:conslaw}) is a consequence of the Yang-Mills equations (\ref{YM}).

In particular, there is a trivial solution of the Yang-Mills equations $B_\mu=0$, $F_{\mu\nu}=0$ for $J^\nu= 0$. If $U=U(x)$ is a scalar field with values in $\G$, then we can get another (vacuum) solution of the Yang-Mills equations using gauge transformation (\ref{gauge:tr}):
$$
B_\mu=-U^{-1}\partial_\mu U,\quad F_{\mu\nu}=0,\quad J^\nu=0.
$$

During the last 60 years several classes of solutions of the Yang-Mills equations were discovered. Namely, monopoles (Wu, Yang, 1968 \cite{WYa}), instantons (Belavin, Polyakov, Schwartz, Tyupkin, 1975 \cite{Bel}), merons (de Alfaro, Fubini, Furlan, 1976 \cite{deA}) and so on (see review of Actor, 1979 \cite{Actor} and review of Zhdanov and Lagno, 2001 \cite{Zhdanov}). Constant solutions of the Yang-Mills equations with zero current are discussed in \cite{Sch} and \cite{Sch2}. 

Let us consider the system of the Yang-Mills equations (\ref{YM}) in the Lie algebra $\frak{g}=\cl^\circledS_{p,q}$,
i.e. $B_\mu\in\cl^\circledS_{p,q}\T_1$, $F_{\mu\nu}\in\cl^\circledS_{p,q}\T_2$, $J^\nu\in\cl^\circledS_{p,q}\T^1$.

\begin{thm}\label{CCS} If the covariantly constant tensor field with values in the Clifford algebra $K_\mu\in \M\cl_{p,q}\T_1$ is a solution of the following system of algebraic equations
\begin{eqnarray}
[K_\mu,[K^\mu, K^\nu]]=J^\nu,\qquad \nu=1, \ldots, m,\label{algur}
\end{eqnarray}
for some $J^\mu\in\M\cl_{p,q}\T^1$, then the tensor field
\begin{eqnarray}
B_\mu(x)=C_\mu(x)+K_\mu(x),\qquad \mu=1, \ldots, m,\label{razlo}
\end{eqnarray}
is a solution of the Yang-Mills equations
\begin{eqnarray}
&&\partial_\mu B_\nu-\partial_\nu B_\mu-[B_\mu, B_\nu]=F_{\mu\nu},\qquad \mu, \nu=1, \ldots, m, \label{YM2}\\
&&\partial_\mu F^{\mu\nu}-[B_\mu,F^{\mu\nu}]=J^\nu, \qquad \nu=1, \ldots, m,\nonumber
\end{eqnarray}
in the Lie algebra $\cl^\circledS_{p,q}$, where $C_\mu\in\cl^\circledS_{p,q}\T_1$ is a unique solution of
$$\partial_\mu h^a-[C_\mu, h^a]=0,\qquad \mu=1, \ldots, m,\qquad a=1, \ldots, n.$$
\end{thm}
\begin{proof}
Let us substitute (\ref{razlo}) into the first equation (\ref{YM2}). We have
\begin{eqnarray}
&&\partial_\mu (C_\nu+K_\nu)-\partial_\nu (C_\mu+K_\mu)-[C_\mu+K_\mu, C_\nu+K_\nu]=F_{\mu\nu},\nonumber\\
&&(\partial_\mu C_\nu-\partial_\nu C_\mu-[C_\mu, C_\nu])+(\partial_\mu K_\nu-[C_\mu, K_\nu])\nonumber\\
&&-(\partial_\nu K_\mu-[C_\nu, K_\mu])-[K_\mu, K_\nu]=F_{\mu\nu}.\nonumber
\end{eqnarray}
All summands, except the last one $[K_\mu, K_\nu]$, equal zero. Using (\ref{CC}) and $K_\nu\in\M\cl_{p,q}\T_1$, we get
$$-[K_\mu, K_\nu]=F_{\mu\nu}.$$
Let us substitute this expression and (\ref{razlo}) into the second equation (\ref{YM2}). We have
\begin{eqnarray}
&&-\partial_\mu [K^\mu, K^\nu]+[C_\mu+K_\mu,[K^\mu, K^\nu]]=J^\nu,\nonumber\\
&&[K_\mu, [K^\mu, K^\nu]]-(\partial_\mu [K^\mu, K^\nu]-[C_\mu, [K^\mu, K^\nu]])=J^\nu.\nonumber
\end{eqnarray}
The expression in round brackets equals zero because $[K^\mu, K^\nu]\in \M\cl_{p,q}\T^2$ and we conclude that $K_\mu$ must satisfy the following equation
$$[K_\mu, [K^\mu, K^\nu]]=J^\nu.$$
The theorem is proved.
\end{proof}

We call the solutions of the Yang-Mills equations from Theorem \ref{CCS} \emph{covariantly constant solutions} because expressions $K_\mu=K_\mu(x)$ are covariantly constant tensor fields with values in Clifford algebra.

In Section \ref{sec5}, we discuss another statements (including the main result of the paper \cite{marchuk1}) which are particular cases of Theorem \ref{CCS}.

\section{Statements in the algebra of $h$-forms}\label{sec5}

In the previous sections of this paper, we consider Clifford algebra $\cl_{p,q}$, $p+q=n$ with matrix $\eta$ (\ref{eta}) and pseudo-Euclidean space $\R^{k,l}$, $k+l=m$ with matrix $\rho$ (\ref{rho}).

Now let us consider the particular case $k=p$, $l=q$. We have $m=n=p+q$ and
$$\eta=\rho=\diag(\underbrace{1, \ldots, 1}_{p}, \underbrace{-1, \ldots, -1}_{q}).$$
We consider a vector field with values in the Clifford algebra $h^\mu=h^\mu(x):\R^{p,q}\to\cl_{p,q}$ (we write $h^\mu\in\cl_{p,q}\T^1$)
\begin{eqnarray}
h^\mu(x)=y^\mu(x) e +y^\mu_a(x) e^a+y^\mu_{ab}(x)e^b+\cdots+ y^\mu_{1\ldots n}(x)e^{1\ldots n}=y^\mu_{A}e^A,\label{hmue}
\end{eqnarray}
which satisfies
\begin{eqnarray}
h^\mu(x) h^\nu(x)+h^\nu(x) h^\mu(x)=2\eta^{\mu\nu}e,\qquad \forall x\in\R^{p,q}.\label{antikom}
\end{eqnarray}
In the case of odd $n$, we also require additional condition
$$\Tr(h^1(x) \ldots h^n(x))=0\qquad \mbox{(in the case of odd $n$)}$$
to obtain independent elements $h^{\mu_1 \ldots \mu_k}$.
The expression $h^\mu$ is called \emph{Clifford field vector} (see \cite{marchuk1}, \cite{ofe}).
The following expression
$$U=ue+u_{\omega_1} h^{\omega_1}+u_{\omega_1 \omega_2}h^{\omega_1 \omega_2}+\cdots+ u_{1 \ldots n}h^{1\ldots n}=u_{\Omega}h^\Omega,$$
where $u_\Omega=u_{\omega_1 \ldots \omega_j}$ are skewsymmetric tensor fields of rank $j$, is called \emph{$h$-form}. The set of such $h$-forms is \emph{an algebra of $h$-forms} $\cl[h]_{p,q}$. It is a generalization of Atiyh-K\"ahler algebra \cite{graf}, \cite{salingaros}, \cite{Marchukeng}, where we have differentials $dx^\mu$ instead of $h^\mu$. The set $h^\mu$, $\mu=1, \ldots, n=p+q$ generate a basis of $\cl[h]_{p,q}$:
\begin{eqnarray}
\{h^\Omega, |\Omega|=0, 1, \ldots, n\}=\{e, h^{\omega_1}, h^{\omega_1\omega_2}, \ldots, h^{1\ldots n}\}.\label{hfo}
\end{eqnarray}
Consider an arbitrary tensor field with values in the algebra of $h$-forms $U^\Phi_\Psi=U^{\phi_1 \ldots \phi_r}_{\psi_1 \ldots \psi_s}(x):\R^{p,q}\to\cl[h]_{p,q}$:
$$U^\Phi_\Psi=U^{\phi_1 \ldots \phi_r}_{\psi_1 \ldots \psi_s}=u^{\phi_1 \ldots \phi_r}_{\psi_1 \ldots \psi_s} e+u^{\phi_1 \ldots \phi_r}_{\psi_1 \ldots \psi_s \omega_1} h^{\omega_1}+ u^{\phi_1 \ldots \phi_r}_{\psi_1 \ldots \psi_s\omega_1 \omega_2}h^{\omega_1 \omega_2}+\cdots$$
$$\cdots+ u^{\phi_1 \ldots \psi_r}_{\psi_1 \ldots \psi_s 1 \ldots n}h^{1\ldots n}=u^{\Phi}_{\Psi \Omega}(x)h^\Omega(x)\in\cl[h]_{p,q}\T^r_s.$$
Note that using (\ref{hmue}), we get $h^\Omega(x)=y^\Omega_A(x) e^A$ for some $y^\Omega_A=y^\Omega_A(x)$ and
$$U^\Phi_\Psi=u^{\Phi}_{\Psi \Omega}h^\Omega=u^{\Phi}_{\Psi \Omega}y^\Omega_A e^A=u^{\Phi}_{\Psi A} e^A\in\cl_{p,q}\T^r_s,\quad u^\Phi_{\Psi A}(x)= u^{\Phi}_{\Psi \Omega}(x)y^\Omega_A(x).$$
Note that the Clifford field vector $h^\mu\in\cl_{p,q}\T^1$ can be regarded as a vector field with values in the algebra of $h$-forms because $h^\mu=\delta^\mu_\nu h^\nu\in\cl[h]_{p,q}\T^1$.

We have analogues of Theorems \ref{th1} - \ref{th4} (see \cite{ofe}), \ref{th5} - \ref{th6}, \ref{th7} - \ref{th9} not for elements $h^a$, but for elements $h^\mu$. Namely, from (\ref{antikom}) it follows that $h^\mu\in\cl^\circledS_{p,q}\T^1$. Instead of equation (\ref{prim}) we have the following equation in the Lie algebra $\cl[h]^\circledS_{p,q}=\cl[h]_{p,q}\setminus \Cen(\cl[h]_{p,q})$:
\begin{equation}
\partial_\mu h^\rho-[C_\mu, h^\rho]=0,\quad \mu,\rho=1,\ldots,n,\label{nik:eq}
\end{equation}
where $h^\rho\in\cl_{p,q}\T^1$ is an arbitrary Clifford field vector and $C_\mu=C_\mu(x)$ ($x\in\R^{p,q}$) is a covector field with values in $\cl_{p,q}$.
The components of the covector field $C_\mu$ satisfy
\begin{equation}
\partial_\mu C_\nu-\partial_\nu C_\mu-[C_\mu,C_\nu]=0,
\qquad\mu,\nu=1,\ldots n.
\end{equation}
The equation (\ref{nik:eq}) is gauge invariant. Let $h^\nu\in\cl^\circledS_{p,q}\T^1$ be a Clifford field vector and $C_\mu\in\cl^\circledS_{p,q}\T_1$ satisfy (\ref{nik:eq}). Let $S : \R^{p,q}\to\cl^\times_{p,q}$ be a function with values in $\cl^\times_{p,q}$ such that
$$
S^{-1}\partial_\mu S\in\cl^\circledS_{p,q}\T_1.
$$
Then the following expressions
$$
\acute h^\rho=S^{-1}h^\rho S\in\cl^\circledS_{p,q}\T_1,\quad
\acute C_\mu=S^{-1}C_\mu S-S^{-1}\partial_\mu S\in\cl^\circledS_{p,q}\T_1
$$
also satisfy the equation
$$
\partial_\mu \acute h^\rho-[\acute C_\mu,\acute h^\rho]=0,\qquad \mu,\rho=1,\ldots,n.
$$
We have a unique solution $C_\mu\in\cl^\circledS_{p,q}\T_1$ of the equation (\ref{nik:eq}):
\begin{equation}
\partial_\mu h^\rho-[C_\mu, h^\rho]=0 \quad\Leftrightarrow\quad C_\mu=\sum_{j=1}^{2[\frac{n}{2}]} \mu_j \pi[h]_j ((\partial_\mu h^\rho) h_\rho),\label{resh}
\end{equation}
where $\mu_j=(n-(-1)^j(n-2j))^{-1}$ and
$$\pi[h]_j:\cl[h]_{p,q}\to\cl[h]^j_{p,q}=\{u_\Omega h^\Omega, |\Omega|=j\}$$
are projection operators. The solution (\ref{resh}) can be also represented in the following form (the proof is similar to the proof of Theorem \ref{th6})
\begin{eqnarray}
C_\mu=\frac{1}{2^n}(\partial_\mu h^\Omega)h_\Omega,\qquad \mu=1, \ldots, n.\label{cmuomega}
\end{eqnarray}
From equation (\ref{nik:eq}) it follows that (the proof is similar to the proof of Theorem \ref{th5})
\begin{eqnarray}
\partial_\mu(h^{\nu_1}\cdots h^{\nu_k})-[C_\mu, h^{\nu_1}\cdots h^{\nu_k}]=0.
\end{eqnarray}
We can consider the operation of covariant differentiation of an arbitrary tensor field $U^{\phi_1 \ldots \phi_r}_{\psi_1 \ldots \psi_s}\in\cl[h]_{p,q}\T^r_s$ with values in the algebra of $h$-forms
$$D_\mu(U^{\phi_1 \ldots \phi_r}_{\psi_1 \ldots \psi_s})=\partial_\mu(U^{\phi_1 \ldots \phi_r}_{\psi_1 \ldots \psi_s})-[C_\mu, U^{\phi_1 \ldots \phi_r}_{\psi_1 \ldots \psi_s}].$$
We conclude that covariant derivative acts as partial derivative acts only on the coefficients before $h^{\omega_1 \ldots \omega_j}$ (the proof is similar to the proof of Theorem \ref{th7}):
\begin{eqnarray}
D_\mu(U^{\phi_1 \ldots \phi_r}_{\psi_1 \ldots \psi_s})=\partial_\mu(u^{\phi_1 \ldots \phi_r}_{\psi_1 \ldots \psi_s})e+\partial_\mu(u^{\phi_1 \ldots \phi_r}_{\psi_1 \ldots \psi_s\omega}) h^\omega+\cdots+ \partial_\mu(u^{\phi_1 \ldots \phi_r}_{\psi_1 \ldots \psi_s 1\ldots n})h^{1\ldots n}.\nonumber
\end{eqnarray}
Let us consider covariantly constant tensor fields with values in the algebra of $h$-forms:
\begin{eqnarray}
\M\cl[h]_{p,q} T^r_s=\{U^{\phi_1 \ldots \phi_r}_{\psi_1 \ldots \psi_s}\in\cl[h]_{p,q}\T^r_s: D_\mu(U^{\phi_1 \ldots \phi_r}_{\psi_1 \ldots \psi_s})=0\}.
\end{eqnarray}
Elements of this set have the form $U^\Phi_\Psi(x)=u^\Phi_{\Psi \Omega}h^\Omega(x)$, where all $u^{\Phi}_{\Psi \Omega}\in\R$ do not depend on $x\in\R^{p,q}$.

\begin{thm}\label{CCS2} If the covariantly constant tensor field with values in the algebra of $h$-forms $K_\mu\in \M\cl[h]_{p,q}\T_1$ is a solution of the following system of algebraic equations
\begin{eqnarray}
[K_\mu,[K^\mu, K^\nu]]=J^\nu,\qquad \mu=1, \ldots, n,\label{algur2}
\end{eqnarray}
for some $J^\mu\in\M\cl[h]_{p,q}\T^1$, then the tensor field
\begin{eqnarray}
B_\mu(x)=C_\mu(x)+K_\mu(x),\qquad \mu=1, \ldots, n\label{razlo2}
\end{eqnarray}
is a solution of the Yang-Mills equations
\begin{eqnarray}
&&\partial_\mu B_\nu-\partial_\nu B_\mu-[B_\mu, B_\nu]=F_{\mu\nu},\qquad \mu, \nu=1, \ldots, n,\label{YM22}\\
&&\partial_\mu F^{\mu\nu}-[B_\mu,F^{\mu\nu}]=J^\nu,\qquad \nu=1, \ldots, n\nonumber
\end{eqnarray}
in the Lie algebra $\cl[h]^\circledS_{p,q}$, where $C_\mu\in\cl[h]^\circledS_{p,q}$ is a unique solution of
\begin{eqnarray}
\partial_\mu h^\nu-[C_\mu, h^\nu]=0,\qquad \mu, \nu=1, \ldots, n.\label{tr6}
\end{eqnarray}
\end{thm}
\begin{proof} The proof is similar to the proof of Theorem \ref{CCS}. The unique solution of the system (\ref{tr6}) is given by (\ref{resh}) or (\ref{cmuomega}).
\end{proof}

Theorem \ref{CCS2} can be regarded as a particular case of Theorem \ref{CCS} because we have $k=p$ and $l=q$ for $\R^{k,l}$ and $\cl_{p,q}$ in Theorem \ref{CCS2}.

In the particular case, elements $h^a$ from Sections \ref{sec1}-\ref{sec4} of this paper (see (\ref{he}), (\ref{hh}), (\ref{trh})) are connected with the generators $e^a$ as $h^a=y^a_b e^b$ using orthogonal matrix $Y=||y^a_b||\in\Or(p,q)$  (see (\ref{hy}) and (\ref{yy})) and elements $h^\mu$ are connected with $e^a$ as $h^\mu=y^\mu_a e^a$ using frame field $y^\mu_a$, $y^\mu_a y^\nu_b \eta^{ab}=\eta^{\mu\nu}$ (see \cite{marchuk1}). In this case, elements $h^a$ and $h^\mu$ are connected as $h^\mu=z^\mu_a h^a$ using frame field $z^\mu_a=y^\mu_b q^b_a$, $z^\mu_a z^\nu_b \eta^{ab}=\eta^{\mu\nu}$, where $Q=||q^b_a||=Y^{-1}$.

As a particular case of Theorem \ref{CCS2} we obtain the following theorem.

\begin{thm}\label{CCS3}
Let $h^\mu\in\cl[h]^\circledS_{p,q}\T^1$ be a Clifford field vector and $C_\mu\in\cl[h]^\circledS_{p,q}\T_1$ satisfies (\ref{nik:eq}). Then the following covector
\begin{equation}
B_\mu=\sigma h_\mu+C_\mu\in\cl^\circledS_{p,q}\T_1,\qquad \mu=1, \ldots, n\label{YM:var}
\end{equation}
is a solution of the following system of Yang-Mills equations:
\begin{eqnarray}
\partial_\mu B_\nu-\partial_\nu
B_\mu-[B_\mu,B_\nu] &=&F_{\mu\nu},\qquad \mu, \nu=1, \ldots, n,\label{YM:eq1}\\
\partial_\mu F^{\mu\nu}-[B_\mu,F^{\mu\nu}] &=& \varepsilon h^\nu,\qquad \nu=1, \ldots, n,\label{YM:eq2}
\end{eqnarray}
where constants $\sigma, \varepsilon\in\R$ are related by the formula
$$
\varepsilon = 4(n-1)\sigma^3.
$$
\end{thm}
\begin{proof} We use Theorem \ref{CCS2} for the current $J^\mu=\varepsilon h^\nu$ and use the formulas
from \cite{MSh}
$$h_\mu h^\nu h^\mu=(2-n) h^\nu,\qquad h_\mu h^\mu=n.$$
We have
\begin{eqnarray}
&&[h_\mu,[h^\mu,h^\nu]]=h_\mu h^\mu h^\nu-h_\mu h^\nu h^\mu-h_\mu h^\nu h^\mu+h^\nu h^\mu h_\mu\nonumber\\
&&=n h^\nu-(2-n)h^\nu-(2-n)h^\nu+n h^\nu=4(n-1)h^\nu.\nonumber
\end{eqnarray}
The theorem is proved.
\end{proof}

As a particular case of Theorem \ref{CCS3} we obtain the main result of the paper \cite{marchuk1}. In the case $h^\mu\in\cl^1_{p,q}\T^1$, we have $h^\mu(x)=y^\mu_a(x) e^a$ for some frame field $y^\mu_a=y^\mu_a(x)$ and the spin connection equals $C_\mu=\frac{1}{4}(\partial_\mu h^\nu)h_\nu\in\cl^2_{p,q}\T^1$.

\begin{thm}\cite{marchuk1} \label{CCS4}
Let $h^\mu\in\cl^1_{p,q}\T^1$ be a Clifford field vector. Then the following covector
\begin{equation}
B_\mu=\sigma h_\mu+\frac{1}{4}(\partial_\mu h^\nu)h_\nu \in(\cl^1_{p,q}\oplus\cl^2_{p,q})\T_1,\qquad \mu=1, \ldots, n\label{YM:var}
\end{equation}
is a solution of the Yang-Mills equations
\begin{eqnarray}
\partial_\mu B_\nu-\partial_\nu
B_\mu-[B_\mu,B_\nu] &=&F_{\mu\nu},\qquad \mu, \nu=1, \ldots, n,\label{YM:eq1}\\
\partial_\mu F^{\mu\nu}-[B_\mu,F^{\mu\nu}] &=& \varepsilon h^\nu,\qquad \nu=1, \ldots, n,\label{YM:eq2}
\end{eqnarray}
where constants $\sigma, \varepsilon\in\R$ are related by the formula
$$
\varepsilon = 4(n-1)\sigma^3.
$$
\end{thm}

Note that the statements of Theorems \ref{CCS3} and \ref{CCS4} can be used in the study of Yang-Mills-Proca equations \cite{YMP}.

\medskip

Note that all considerations of this paper can be reformulated for the case of the complexified Clifford algebra $\C\otimes\cl_{p,q}$ and the corresponding Lie algebra. We can also consider the Lie algebra $\C\otimes\cl^\circledS_{p,q}$. The constants $\varepsilon$ and $\sigma$ in Theorems \ref{CCS3} and \ref{CCS4} will be complex in this case.

We have the following well-known isomorphisms \cite{Lounesto}
\begin{eqnarray}
\C\otimes\cl_{p,q}\cong\left\lbrace
\begin{array}{ll}
\Mat(2^{\frac{n}{2}}, \C), & \parbox{.5\linewidth}{if $n$ is even,}\\
\Mat(2^{\frac{n-1}{2}}, \C)\oplus \Mat(2^{\frac{n-1}{2}}, \C), & \parbox{.5\linewidth}{if $n$ is odd.}
\end{array}\nonumber
\right.
\end{eqnarray}
In the case of even $n$, we have the Lie group isomorphism
\begin{eqnarray}
(\C\otimes\cl_{p,q})^\times\cong\GL(2^{\frac{n}{2}}, \C).
\end{eqnarray}
Taking into account $\U(2^{\frac{n}{2}})\subset \GL(2^{\frac{n}{2}}, \C)$ and using operation of Hermitian conjugation in Clifford algebra \cite{MSh}, we can reformulate theorems of the current paper with the use of unitary Lie groups (and unitary Lie algebras) of corresponding dimensions. We can also use another classical Lie groups and corresponding Lie algebras in the complexified Clifford algebra $\C\otimes\cl_{p,q}$ (see papers \cite{sh1}, \cite{sh2}, \cite{sh3}).

We discuss mathematical structures and constructions in this paper. Relating the proposed mathematical constructions to real word objects goes beyond the scope of this investigation. The application of the methods of this article to other nonlinear equations of mathematical physics is the subject for further research.

\subsection*{Acknowledgments}

The author is grateful to N. G. Marchuk for fruitful discussions
and participants of ICCA 11 Conference for useful comments. The author is grateful to anonymous reviewer for careful reading of the manuscript and helpful comments on how to improve the presentation.

\end{document}